\documentclass[10pt,final,reqno]{amsart}

\usepackage{amsmath,amsfonts,amssymb,amsopn}

\usepackage{setspace}

\usepackage{times}

\usepackage{fancyhdr}

\DeclareMathSymbol{\E}{\mathalpha}{AMSb}{"45}
\DeclareMathSymbol{\R}{\mathalpha}{AMSb}{"52}

\renewcommand{\P}{\mathbb{P}}

\newcommand{\one}{\mathbf{1}}

\newcommand{\N}{\mathbb{N}}
\newcommand{\Z}{\mathbb{Z}}

\newcommand{\F}{\mathcal{F}}

\newcommand{\ST}{\mathfrak{T}}

\newcommand{\STtT}{\ST [t,T]}

\newcommand{\Ar}{\mathfrak{R}}

\renewcommand{\L}{\operatorname{L}}

\DeclareMathOperator{\D}{D}
\newcommand{\Dl}{\D_{\ell}}

\newcommand{\Dlk}{\D_{\ell_k}}

\DeclareMathOperator{\fddelta}{\delta}
\newcommand{\fddeltatau}{\fddelta_\tau}
\newcommand{\fddeltak}{\fddelta_{h_k,\ell_k}}

\DeclareMathOperator{\fdDelta}{\Delta}
\newcommand{\fdDeltak}{\fdDelta_{h_k,\ell_k}}

\DeclareMathOperator{\tr}{tr}

\newcommand{\mt}{\mathcal{M}_T}
\newcommand{\bmt}{\bar{\mt}}

\newcommand{\gb}{g_{R_1}}

\newcommand{\wb}{w_{R_1}}
\newcommand{\wth}{w_{\tau,h}}

\newcommand{\wthRone}{w_{\tau,h}^{R_1}}
\newcommand{\wthRR}{w_{\tau,h}^{R,R_1}}

\newcommand{\on}{\,\, \textrm{on} \,\,}
\newcommand{\cand}{\quad \textrm{and} \quad}

\newcommand{\cfor}{\quad \text{for} \quad}

\newcommand{\e}{\varepsilon}

\newtheorem{anytheorem}{Theorem}[section] 
\newtheorem{theorem}[anytheorem]{Theorem}
\newtheorem{lemma}[anytheorem]{Lemma}
\newtheorem{corollary}[anytheorem]{Corollary}
\theoremstyle{definition}
\newtheorem{assumptions}[anytheorem]{Assumption}
\newtheorem{remark}[anytheorem]{Remark}

\begin{document}
\title{Error estimates for finite difference approximations of American put option price} 

\author{David \v{S}i\v{s}ka} 
\address{Fakult\"at f\"ur Mathematik, Universit\"at Bielefeld, Universit\"atsstra{\ss}e, D-33501 Bielefeld, Germany}
\email{dsiska@math.uni-bielefeld.de}
\date{\today}

\begin{abstract}
Finite difference approximations to multi-asset American put option price are considered. 
The assets are modelled as a multi-dimensional diffusion process with variable drift and volatility. 
Approximation error of order one quarter with respect to the time discretisation parameter and one half with respect to the space discretisation parameter is proved by reformulating the corresponding optimal stopping problem 
as a solution of a degenerate Hamilton-Jacobi-Bellman equation.
Furthermore, the error arising from restricting the discrete problem to a finite grid by reducing the original problem to a bounded domain is estimated.
\end{abstract}

\subjclass{65M06; 65M12; 60G40; 35R35; 91G80; 91G60}
\keywords{American put option, Finite difference method, Optimal stopping, Optimal control, Hamilton-Jacobi-Bellman equation}

\maketitle

\section{Introduction}
\label{s:oa}
American put option is a derivative contract based
on the price of some asset, denoted by $S_t$, which evolves with time. 
At time $t=0$, when the contract is
entered, a reference asset is chosen together with
``strike'' $K_\textrm{strike}$ and expiry time $T$. The contract
gives the holder the right (but not the
obligation) to sell the asset at any time $T^*
\in [0,T]$ for the amount $K_\textrm{strike}$. Thus, when
exercised, the payoff from the American put is $K_\textrm{strike}
- S_{T^*}$. The option will only ever be exercised by a ``rational'' investor
if the payoff is positive, so the payoff from
American put is $[K_\textrm{strike}-S_{T^*}]_+$, where $[x]_+$ denotes the positive part of any real number $x$.
Even in the classical Black-Scholes model, there is no known formula for the price of an American put with a finite exercise time (there are formulae for prices of infinite exercise time American put, American call option and European put and call options). 
A variety of numerical methods and approximations
for the American put option price have been
developed over the years. An overview of the
various methods can be found for example in
Barone-Adesi~\cite{baroneadesi:saga}. 

Four main approaches for calculating the put option price can be identified: 
The first attempts to find formulae that give results close to the real price. 
These give fast approximations but the accuracy cannot 
be simply improved upon by doing more computations. The
second approach approximates the evolution of the 
underlying asset with a recombining tree (typically
binomial or trinomial). Backward induction on the tree
then yields the American put price. 
Error estimates for the binomial tree approach are 
proved in Lamberton~\cite{lamberton:error} and improved in 
Lamberton and Rogers~\cite{lamberton:optimal} and Lamberton~\cite{lamberton:brownian}. 
The third approach is based on Monte Carlo methods.
Finally the fourth approach relies on reformulating the
option price as a solution to a partial differential inequality or a nonlinear partial differential equation. This equation can then be discretised using a variety of methods. For the finite element
method see for example Allegretto et al.~\cite{allegretto:finite} and also Pironneau and Achdou 
\cite[Section 6.4]{bensoussan:handbook}.
Finite volume methods have been used by Angermann and Wang~\cite{angermann:convergence} and Berton and Eymard~\cite{berton:eymard:finite}. 
Nevertheless finite difference methods are particularly popular and it is those that we focus on here. 

If a diffusion process is used to model 
the asset price, then it has been shown that the
American put price is the payoff function of an
optimal stopping problem. 
See for example Shiryaev~\cite{shiryaev:essentials}.
The payoff function for this optimal stopping
problem satisfies a system of second order partial
differential inequalities. The solution to this
system can be approximated using finite difference
methods. See for example Lamberton and Lapeyre
\cite[Chapter 5]{lamberton:introduction} for a concise introduction.
They also present a simple algorithm for computing 
the solution to the finite difference problem. 
A more efficient iterative method based on the SOR 
method, called projected SOR, is given in Pironneau and Achdou
\cite[Section 6.4.1]{bensoussan:handbook}.
Finding more efficient ways of computing 
the solution of the finite difference
problem are of considerable interest. 
The reader is referred to Forsyth and Vetzal~\cite{forsyth:quadratic}
and Cen and Le~\cite{cen:robust}
where the penalty method
and singularity separating, 
implicit finite difference scheme, 
are studied and compared with other methods.

This paper is focused on implicit finite difference approximations to the American put option price in the case when there is one or more underlying assets with variable diffusion coefficient, drift, and discounting.  
We prove that the error introduced by the implicit finite difference approximation is, under suitable regularity assumptions, of order $\tau^{1/4}+h^{1/2}$, where $\tau$ denotes the space discretisation parameter and $h$ denotes the space discretisation parameter. 
The only other result in this direction is Hu et. al.~\cite{hu:optimal}, as far as the author is aware. In Hu et. al.~\cite{hu:optimal} 
rate of convergence of order $\tau^{1/2} + h$ is proved and furthermore it is show that this is an optimal rate of convergence.
However this is done in a much simpler setting then this article considers. In particular, only one risky asset is considered and it is assumed to have constant drift and diffusion coefficients. Furthermore this article adds the estimates for the error arising in computing the discrete problem on a finite domain.
The reason why the same rate of convergence as in Hu et. al.~\cite{hu:optimal} is not obtained in this article is that here the diffusion coefficients are allowed to degenerate.
Hence the solution can only be Lipschitz continuous in the space variable.
The result of Hu et. al.~\cite{hu:optimal} requires more regularity and only holds in the non-degenerate case.

The way to obtaining the rate of convergence has been paved by recent work of Krylov on the rates of convergence of finite difference approximations to the Bellman equation (also known as Hamilton-Jacobi-Bellman equation).
The first rate of convergence estimates were obtained in
Krylov~\cite{krylov:rate:equations}.
This has later been
extended to the case of variable coefficients in Krylov~\cite{krylov:rate:variable}.
The rate of convergence has been further improved, in the case of constant coefficients, in Krylov and Dong,~\cite{dong:krylov:rate:constant}. 
Finally the convergence rate of order $\tau^{1/4} + h^{1/4}$ for Bellman equations with Lipschitz continuous coefficients is obtained in
Krylov~\cite{krylov:rate:lipschitz:published}. 
These results have been extended in Gy\"ongy and \v{S}i\v{s}ka~\cite{gyongy:siska:on:the:rate} 
to cover the Bellman equation corresponding to optimal stopping of controlled diffusion with Lipschitz continuous coefficients. In Krylov and Dong~\cite{dong:krylov:on:time:inhomogenous} the results from Krylov~\cite{krylov:rate:lipschitz:published} have been extended to allow domains not equal to $\R^d$.
In Krylov~\cite{krylov:apriori} one of the key ingredients of the proof, the discrete gradient estimate, has been generalised to allow estimates for other nonlinear partial differential equations. 
The constant coefficients case has been studied, 
also adapting Krylov's methods, in Jakobsen
\cite{jakobsen:rate:optimal:stopping}. Krylov~\cite{krylov:rate:lipschitz:published} considers general finite difference schemes which have been already introduced in Bonnans and Zidani~\cite{bonnans:consistency} in the controlled Markov chain setting, however without establishing any rates of convergence. 

The problem of restricting partial differential equations arising in finance to bounded domains has been studied in Barles, Daher and Romano~\cite{barles:daher:romano:convergence}. 
By introducing artificial boundary condition of either Dirichlet or Neumann type on the boundary of the ball to which they restrict the domain they are able to prove convergence to the solution of the equation on the whole space.
This applies to a class of nonlinear partial differential equations to which the viscosity solutions exist.
However only in the case of linear equations with constant coefficients do they get an exponential rate of convergence.
This paper demonstrates exponential rate of convergence for a nonlinear problem (the American option price) with variable coefficients. This is proved by first obtaining a general result on the distribution of exit times of a diffusion process from a ball, see Lemma~\ref{lemma-exp-est}, and second by applying the maximum principle for the discretised equation.

The paper is organised as follows. 
In Section~\ref{s:mr} the main result is presented together with the assumptions.
In Section~\ref{s:os} the partial differential equation for the American put option price is obtained together with the rate of convergence estimate for approximations on infinite grids. In Section~\ref{sec:cyl:dom} the error arising from restricting the infinite grid to a finite grid is estimated.

\section{Main result}                                           \label{s:mr} 
Let $(\Omega,\F, \P, (\F_t)_{t\geq0})$ be a probability space
with a right-continuous filtration, such that
$\F_0$ contains all $\P$ null sets.  Let
$(W_t,\F_t)$ be a $d$-dimensional Wiener
martingale, i.e., let $(W_t)_{t \geq
  0}$ be adapted to $(\F_t)_{t \geq 0}$ and for
all $t,s \geq 0$, $W_{t+s} - W_t$ independent
of $\F_t$. 

We will consider the standard Black-Scholes model
extended to several dimensions. We consider $d$
risky assets and one risk-less asset. We assume
that we are given $\bar{\rho}=\bar{\rho}(t,x)$, a non-negative real
valued function of $t \in [0,T]$ and $x\in\R^d$, 
representing the continuously 
compounded zero coupon rate and also
$\bar{\sigma}=\bar{\sigma}(t,x)$, a $d\times d$ matrix valued function
of  $t \in [0,T]$ and $x\in\R^d$, representing volatilities 
of the risky assets and the correlations between the risky
assets.
\begin{assumptions}                                           
\label{ass:bnd} 
The functions $\bar{\rho}$ and $\bar{\sigma}$ are Borel measurable in $t$. 
There exists a positive constant $K$ such that 
  \begin{equation}
    |\bar{\sigma}| = (\tr \bar{\sigma} \bar{\sigma}^T)^{1/2} 
    \leq K \cand \bar{\rho} \leq K \on [0,T]\times \R^d.
  \end{equation}
\end{assumptions}
We denote the risky assets $S_u = (S^1_u,\ldots S^d_u)$, where $S_u = (S_u^{t,S})_{u\in [t,T]}$ is defined on $[t,T]$ by the stochastic differential equation
\begin{equation}                                              \label{eq:st:pr}
  dS_u^i = S_u^i \bar{\rho}(u,S_u) du 
    + S_u^i\sum_{j=1}^d\bar{\sigma}^{ij}(u,S_u)dW^j_u,
  \quad S^i_t = S^i, \quad i = 1,\ldots, d,
\end{equation}
where $S > 0$. It is well known (see e.g. Krylov \cite[Chapter 2, Section 5]{krylov:controlled}) that the stochastic
differential equation has a unique solution under
Assumption~\ref{ass:bnd} together with the assumption that the functions $S\mapsto S\bar{\sigma}(u,S)$ and $S\mapsto S\bar{\rho}(u,S)$ are Lipschitz continuous for all $u\in[0,T]$. 
We will use the notation $\E_{t,S}$ to denote the expectation of the expression following with the understanding that the relevant stochastic process is started from point $S$ at time $t$.
We consider the optimal stopping problem
\begin{equation}  
\label{eq:am:op:pr}                                          
  v(t,S) = \sup_{T^* \in \STtT}\E_{t,S}
  \left(e^{-\int_t^{T^*} \bar{\rho}(u,S_u) du} \bar{g}\left(S_{T^*}\right) \right),
\end{equation}
for a given Lipschitz continuous function $\bar{g}$. 
The American put option price is given by $v$. 
See e.g. Shiryaev~\cite{shiryaev:essentials}. 
In the one dimensional case $\bar{g}(S):=[K_\textrm{strike}-S]_+$ but in the multidimensional case one may want to consider a general payoff $\bar{g}$.
We wish to remove the linear growth present in the coefficients of \eqref{eq:st:pr}. 
Hence for $u \in [t,T]$, we define $x^i_u := \ln S^i_u$,
\begin{equation}                                                 \label{eq:ab}
  \bar{\beta}^i(t,S) := \bar{\rho}(t,S) - \tfrac{1}{2}
  \sum_{j=1}^d \bar{\sigma}^{ij}(t,S)^2
\end{equation}
and $\sigma(t,x) = \bar{\sigma}(t,e^x)$, $\beta(t,x) = \bar{\beta}(t,e^x)$ and
$\rho(t,x) = \bar{\rho}(t,e^x)$.
By It\^o's formula we get 
\begin{equation*}
  dx^i_u = \beta^i(u,x^i_u)du 
  + \sum_{j=1}^d\sigma^{ij}(u,x^i_u) dW^j_u, \quad x^i_t = x^i = \ln S^i, \quad i = 1,\ldots, d, \quad u\in[t,T].
\end{equation*}
Let $g(x) := \bar{g}(e^{x_i})$ and $x = \ln S$. 
Then the option value $v(t,S)$ given by \eqref{eq:am:op:pr} is equal to 
\begin{equation}                                           \label{eq:am:op:pr2}
w(t,x) = \sup_{T^* \in \ST[0,T-t]} \E_{t,x}(e^{-\int_0^{T^*}\rho(u,x_u) du }g(x_{T^*})).
\end{equation}
Let $\eta=\eta(x)$ be a smooth function and define 
\begin{equation}                                                   \label{eq:L}
  \L \eta :=  \sum_{i,j=1}^d\tfrac{1}{2}(\sigma\sigma^T)^{ij} \eta_{x^i x^j} +  \sum_{i=1}^d\beta^{i} \eta_{x^i} - \rho \eta.
\end{equation}
\begin{assumptions}                                    \label{ass-on-the-scheme}
There exist a natural number $d_1$, vectors $\ell_k \in \R^d$ and functions
\begin{equation*}
a_k:[0,T]\times \R^d \to \R, \quad b_k:[0,T]\times \R^d  \to \R, \quad k = \pm 1, \ldots, \pm d_1,
\end{equation*}
such that $\ell_k = -\ell_{-k}$, $|\ell_k| \leq K$, $a_k = a_{-k}, b_k \geq 0$ for $k = \pm 1, \ldots, \pm d_1$ and such that
\begin{equation*}
\tfrac{1}{2}(\sigma\sigma^T)^{ij}(t,x) = a_k(t,x) \ell_k^i \ell_k^j, \cand \beta^i(t,x) =  b_k(t,x) \ell_k^i
\end{equation*}
for all $k \in \{\pm 1, \ldots, \pm d_1\}$, $i,j\in\{1,\ldots,d\}$ and $(t,x) \in [0,T]\times \R^d$.
\end{assumptions}
Let $\Dl$ and $\Dl^2$ denote the first and second
derivatives in the direction of a vector $\ell$ in
$\R^d$. Notice that under this assumption $\L$ given by
\eqref{eq:L} satisfies
\begin{equation}\label{eqn-Lk}
  \L \eta = \sum_{k=\pm 1, \ldots, \pm d_1}\left(a_k \Dlk^2 \eta + b_k \Dlk \eta\right) - \rho \eta
\end{equation}
\begin{remark}
While this assumption may
appear restrictive, it turns out that this can be satisfied for any operator given by \eqref{eq:L}, such that:
\begin{enumerate}
\item We can find $\L_h \eta(x) = \sum_{y\in B} p_h(y)\eta(x+hy)$
with a finite $B \subset \R^d$ such that $\textrm{span } B = \R^d$.
\item We have $\L \eta(0) = \lim_{h\searrow 0}\L_h \eta(0)$.
\item If $\eta$ has a strict maximum at $0$ then $\L_h \eta(0) < 0$.
\end{enumerate}
For a proof see Dong and Krylov
\cite[Section 3]{dong:krylov:rate:constant}.
In the one dimensional case the construction of $a_k$ and $b_k$ is straightforward. In several dimensions in the case when $\tfrac{1}{2}\sigma \sigma^T$ is diagonally dominant see e.g. \v{S}i\v{s}ka \cite[Example 5.2.5]{siska:numerical}. See Bonnans and Zidani~\cite{bonnans:consistency} for the case when $\tfrac{1}{2}\sigma \sigma^T$ is not diagonally dominant.
\end{remark}
We will need the following regularity assumptions.
\begin{assumptions}                                  
\label{ass-holder}
For $\psi \in \{\sigma, \sqrt{a_k}, b_k, \rho$  for $k=\pm 1, \ldots, \pm d_1\}$ there exists $K>0$ such that for all $t,s \in [0,T)$ and for all $x,y \in \R^d$
\begin{equation*}
|\psi(t,x) - \psi(s,x)| \leq K|t-s|^{1/2},\quad |\psi(t,x)| \leq K. \cand |\psi(t,x) - \psi(t,y)| \leq K|x-y|.
\end{equation*}
\end{assumptions}
Notice that Assumptions~\ref{ass-on-the-scheme} and 
\ref{ass-holder} imply Assumption~\ref{ass:bnd}.
Let $h > 0$, $\tau > 0$ and $\ell \in \R^d$. 
We define $\tau_T(t) := \min(\tau, T-t)$. 
So the time step is fixed except the case $t\in(\tau,T)$ when $T-t$ is used instead of $\tau$. 
Let
\begin{equation}                                              \label{eq:fdops}
\begin{split}
\fddeltatau^T \eta(t,x) & := \frac{\eta(t+\tau_T(t),x) - \eta(t,x)}{\tau},\quad
\fddelta_{h,\ell} \eta(t,x) := \frac{\eta(t,x+h\ell) - \eta(t,x)}{h},\\
\fdDelta_{h,\ell} \eta & := - \fddelta_{h,\ell} \fddelta_{h,-\ell} \eta = \frac{1}{h}(\fddelta_{h,\ell} \eta + \fddelta_{h,-\ell}\eta)
\end{split}
\end{equation}
for $t\in[0,T), x\in \R^d$. 
Let
\begin{equation}                                                  \label{eq:Lh}
\L_h \eta := \sum_{k=\pm 1, \ldots, \pm d_1}\left(a_k \fdDeltak \eta + b_k \fddeltak \eta\right) - \rho \eta,
\end{equation}
where $a_k$ and $b_k$ are the functions from
Assumption~\ref{ass-on-the-scheme}. 
\begin{remark}
Contrary to the usual finite difference approach we have not yet introduced any grid on which the solution to the discrete problem is defined.
Typically, one first introduces a grid and then the finite difference operators acting on functions on the grid.
Here the opposite approach is taken.
First the finite difference operators are defined for any point $(t,x)$. 
Thus we will obtain a collection of disjoint problems, each centered around an arbitrary $(t_0,x_0)\in [0,T]\times \R^d$ and solved on the grid $\mt := \bmt \cap \big([0,T) \times \R^d\big)$, where 
\begin{equation*}
\begin{split}
\bar{\mathcal{M}}_T & := \{(t,x) \in [0,T]\times \R^d : (t,x) = ((t_0+j\tau) \wedge T, \\ 
&\quad \quad x_0+h(i_1\ell_1 + \cdots + i_{d_1} \ell_{d_1})), j \in \{0\} \cup \N, i_k \in \Z, k = \pm 1, \ldots, \pm d_1\}.
\end{split}
\end{equation*}
\end{remark}
The discrete problem to be solved in order to approximate the price of the American put option is
\begin{equation}                                   \label{eq:fd:scheme:american}
\begin{split}
\max\left[\fddeltatau^T \wth + \L_h \wth, g - \wth \right] & = 0 \on Q,\\
\wth & =  g \on \bmt \setminus Q,
\end{split}
\end{equation}
where $Q \subset \mt$. The solution $\wth$ to \eqref{eq:fd:scheme:american} can be defined for any point in $[0,T)\times \R^d$, since the grid $\mt$ can be centered arbitrarily.
\begin{remark}
It is worth noting the nonlinear structure of the above partial differential equation. 
If the equation was linear and non-degenerate, i.e. if, for example, on the left hand side we only had the first term of the maximum and non-degenerate $\L$, we would be in the standard situation of linear parabolic equations and we would immediately know that what the rate of convergence is, from, for example, Thom\'ee~\cite{handbookofnumericalanalysis:vol1:pt1}.
\end{remark}

Finally, we consider the localisation error. 
We will use $R > R_1 > R_2 > 0$. Let $B_R = \{x\in \R^d: |x| < R\}$. We will need to solve the discrete problem on a grid in $B_R$.
The radius $R_1$ is used for  introducing the artificial boundary conditions, while in $B_{R_2}$ we obtain the desired estimate.
Let $g_{R_1}$ be a function that is equal to $g$ inside $B_{R_1}$, zero outside $B_{R_1+1}$ and Lipschitz continuous.
Let $Q_R := ([0,T) \times B_R) \cap \mt$. 
The discrete problem that needs to be solved is 
\begin{equation}                                   \label{eq:fd:scheme:american:2}
\begin{split}
\max\left[\fddeltatau^T \wthRR + \L_h \wthRR, g - \wthRR \right] & = 0 \quad \on Q_R,\\
\wthRR & =  g_{R_1} \on \bmt \setminus Q_R.
\end{split}
\end{equation}
The following theorem is the main result of this paper. We will always use $C>0$ to denote a generic constant that is independent of  $\tau,h, R_1, R_2$ and $R$.
\begin{theorem}                                         \label{thm:american} 
Let Assumptions~\ref{ass-on-the-scheme} and~\ref{ass-holder} be satisfied. Then the system \eqref{eq:fd:scheme:american} has a unique solution for any $Q\subset \mt$ and for any $g=g(x)$ which is bounded and Lipschitz continuous in $x$.
Furthermore, let $w$ be given by \eqref{eq:am:op:pr2}.
Let $\wthRR$ denote the solution to \eqref{eq:fd:scheme:american:2}. 
Then there are constants $\mu > 0$ and $\gamma \in (0,1)$ such that
\begin{equation*}
|w - \wthRR| \leq C(e^{-\mu R_1^2 + R_2^2/2} + \tau^{1/4} + h^{1/2} + e^{\gamma(R_1 - R)}) \on [0,T] \times B_{R_2}.
\end{equation*}
\end{theorem}
The proof will be given in Section~\ref{sec:cyl:dom}.
An outline is given below.
We start by using randomised stopping to express the payoff function of the optimal stopping problem as a payoff of an optimal control problem with unbounded reward and discounting functions. 
The payoff of this optimal control problem then corresponds to the solution of a normalised Bellman equation which we use to derive the finite difference approximation. 
Adapting results from Gy\"ongy and \v{S}i\v{s}ka~\cite{gyongy:siska:on:the:rate} gives the first main result of this paper, the rate of convergence of order $\tau^{1/4} + h^{1/2}$ for a grid on the whole space. 
We then prove an estimate the probability that a stochastic process exits a certain ball before time $T$ and use this together with a discrete comparison principle to estimate the error arising in restricting the approximation to a finite grid.

\section{Normalised Bellman equation}                                                              
\label{s:os}
This section applies known results about optimal stopping, optimal control and normalised Bellman equations to estimate the rate of convergence. 
The following theorem is a special case of Gy\"ongy and \v{S}i\v{s}ka  \cite[Theorem 3.2]{gyongy:siska:on:randomized}. It is also
proved in Krylov \cite[Chapter 3]{krylov:controlled}.
\begin{theorem}                                               \label{t:rs}
Let $\Ar_n$ contain all progressively measurable, locally integrable processes $r = (r_t)_{t\geq 0}$ taking values in $[0,n]$, such that $\int_0^\infty r_t dt = \infty$. 
Let $\Ar = \bigcup_{n \in \N} \Ar_n$. Let $g$ be a Lipschitz continuous function of $x$. 
Then for all $(t,x) \in [0,T]\times \R^d$,
\begin{equation*}
w(t,x) = \sup_{r \in \Ar} \E_{t,x} \left(\int_t^T r_sg(x_s)e^{-\int_t^s \rho(u,x_u) + r_u du}ds + g(x_T)e^{-\int_t^T \rho(u,x_u) + r_udu} \right).
  \end{equation*}
\end{theorem}
\begin{remark}
\label{rem:normBellmaxPde}
Let $\L$ be the differential operator given by
\eqref{eq:L}. Krylov \cite[Theorem 6.3.3]{krylov:controlled} proves that $w$ is the unique solution of the following normalised Bellman equation
\begin{align}
\label{eq:4}
\sup_{r \in R^+} \left( \tfrac{1}{1+r} \left(w_t + \L w\right) + \tfrac{r}{1+r}(g-w) \right) & = 0 \quad \quad \on [0,T)\times \R^d,\\
 w(T,x) & = g(x) \cfor x \in \R^d.
\end{align}
Let $\e =
\tfrac{1}{1+r}$ and use this in \eqref{eq:4}. 
The supremum is now taken over $\e \in
[0,1]$ and hence
\begin{equation*}
\sup_{\e \in [0,1]} \left( \e\left(w_t + \L w\right) + (1-\e)(g-w) \right) = 0 \on [0,T)\times \R^d,
\end{equation*}
Noticing that for any real numbers $p$ and $q$, 
\begin{equation*}
  \sup_{\e \in [0,1]}(\e p + (1-\e)q) = \max(p,q)
\end{equation*}
we obtain that \eqref{eq:4} is equivalent to 
\begin{equation*}
  \max\left[w_t + \L w, g-w\right] = 0\on [0,T)\times \R^d.
\end{equation*}
The reader will immediately recognise that this is 
exactly the nonlinear equation of which \eqref{eq:fd:scheme:american}
is the finite difference approximation.  
\end{remark}

\begin{lemma}                                           \label{l:u}
Let $Q$ be a subset of $\mt$ and $g$ be a bounded and Lipschitz continuous function. Let Assumptions~\ref{ass:bnd}, ~\ref{ass-on-the-scheme} and~\ref{ass-holder} be satisfied.
Then \eqref{eq:fd:scheme:american} has a unique solution.
\end{lemma}
\begin{proof}
We wish to transform \eqref{eq:fd:scheme:american} into the same form as it 
appears in Gy\"ongy and \v{S}i\v{s}ka~\cite{gyongy:siska:on:the:rate}.
Using the same argument as in Remark~\ref{rem:normBellmaxPde} we can see that the first equation in \eqref{eq:fd:scheme:american} is equivalent to 
\begin{equation}                                             \label{eq:5.4}
\sup_{r\in [0,\infty)}\frac{1}{1+r}\left(\fddeltatau^T \wth + \L_h \wth -r\wth + rg \right) = 0 \on Q.
\end{equation}
Now we just need to check that the assumptions required
by Gy\"ongy and \v{S}i\v{s}ka \cite[Theorem 3.4]{gyongy:siska:on:the:rate}
are  satisfied in our case. The space of control parameters
is $[0,\infty)$. The normalising factor $m$ is $1/(1+r)$. 
Assumption 3.1 of~\cite{gyongy:siska:on:the:rate} is
 satisfied due to Assumption~\ref{ass-on-the-scheme}. 
The discounting function $c$ is $\rho + r$ and the reward 
function $f$ is $rg$. Due to Assumption~\ref{ass:bnd} we have
\begin{equation*}
  |\tfrac{1}{1+r} a_k| + |\tfrac{1}{1+r}b_k| 
+ |\tfrac{1}{1+r}(\rho +r)| + |\tfrac{1}{1+r} rg| \leq C.
\end{equation*}
Thus~\cite[Assumption 3.2]{gyongy:siska:on:the:rate} is satisfied. 
Finally $\tfrac{1}{1+r}(1+\rho+r)\geq 1,$ since $\rho\geq 0$. 
Hence \cite[Assumption~3.3]{gyongy:siska:on:the:rate} is satisfied and the result follows from~\cite[Theorem~3.4]{gyongy:siska:on:the:rate}.
\end{proof}

\begin{theorem}                                        
\label{t:rate}
Let $w$ be the American put option price given by \eqref{eq:am:op:pr}. Let $Q = \mt$ and let $\wth$ be the solution of \eqref{eq:fd:scheme:american}. 
Then
\begin{equation*}
|w - \wth| \leq C(\tau^{1/4} + h^{1/2}) \on \mt.
\end{equation*}
\end{theorem}

\begin{proof}
As before we consider the first equation in \eqref{eq:fd:scheme:american}
rewritten as \eqref{eq:5.4}. We now just need to check 
that~\cite[Assumptions 2.3 through 2.5]{gyongy:siska:on:the:rate}
hold, so that we can apply~\cite[Theorem 2.4]{gyongy:siska:on:the:rate}. Notice that~\cite[Assumption 2.3]{gyongy:siska:on:the:rate}
is equivalent to Assumption~\ref{ass-on-the-scheme} here and \cite[Assumption 2.4]{gyongy:siska:on:the:rate} is satisfied due to  Assumptions~\ref{ass:bnd} here together with
and~\ref{ass-holder}. Finally~\cite[Assumption 2.5]{gyongy:siska:on:the:rate} is satisfied thanks to
our Assumption~\ref{ass-holder}. Hence we get the desired rate
of convergence.
\end{proof}

\section{Approximations in cylindrical domains}      \label{sec:cyl:dom}
In this section we estimate the error arising in the
localisation. To this end we prove a general result
on the distribution of the exit time of diffusion processes form
balls of some radius $R$. 
This, together with a comparison type theorem for the finite difference schemes from~\cite{gyongy:siska:on:the:rate} allows us to estimate the localisation error.

\begin{lemma}\label{lemma-exp-est} 
Let $W_t$ be a $d'$-dimensional Wiener martingale. Let $\xi$ be a $\F_0$ measurable, $\R^d$ valued random variable. 
Let $(\sigma_t)_{t\in [0,T]}$ and $(\beta_t)_{t\in [0,T]}$ be $(\F_t)_{t\in [0,T]}$ progressively measurable processes, where $\sigma_t$ is a $d\times d'$ dimensional matrix and $\beta_t$ is a $d'$-dimensional vector such that $\int_0^T |\beta_s| ds < \infty$ and $\E \int_0^T |\sigma_s|^2 ds < \infty$.
Let $(x_t)_{t \in [0,T]}$ be given by
\begin{equation*} 
dx_t = \beta_t dt + \sigma_t dW_t,\quad x_0 = \xi.
\end{equation*}
Let $a_s = \tfrac{1}{2}\sigma_s \sigma_s^T$. Let $(\cdot, \cdot)$ denote the inner product in $\R^d$. 
If for all $t \in [0,T]$ 
\begin{equation}\label{eq:exp:est:ass}
2x_t\beta_t + |\sigma_t|^2 + (a_t x_t,x_t) \leq K(1+|x_t|^2)\quad \P-\textrm{almost surely}
\end{equation}
and if $\E e^{\xi^2} < \infty$, then there exists $\mu > 0$, depending only on $K$ and $T$, such that
\begin{equation}\label{eq:exp:est}
\E \sup_{t \in [0,T]} e^{\mu x_t^2}\leq 3(1+\E e^{\xi^2/2}).
\end{equation}
\end{lemma}

\begin{proof} 
We begin by applying It\^o's formula to the process $x_t$. Thus
\begin{equation*}
dx_t^2 = 2x_t \sigma_t dW_t + (2(x_t, \beta_t) + |\sigma_t|^2)dt
\end{equation*}
and hence
\begin{equation*}
\begin{split}
de^{x_t^2} = & e^{x_t^2}\big[2x_t \sigma_t dW_t + (2(x_t,\beta_t) + |\sigma_t|^2)dt+ 2(a_t x_t,x_t)dt\big].
\end{split}
\end{equation*}
Let $\psi_t := \exp(e^{-\lambda t} x_t^2)$, where $\lambda > 0$ is a constant to be chosen later. 
Then
\begin{equation*}
\begin{split} 
d\psi_t = e^{-\lambda t}\psi_t \big[&2x_t \sigma_t dW_t + (2(x_t, \beta_t) + |\sigma_t|^2 + (a_t x_t,x_t))dt - \lambda x_t^2 dt \big].
\end{split}
\end{equation*} 
Using~\eqref{eq:exp:est:ass} we see that
\begin{equation*}
\psi_t \leq \psi_0 + \int_0^t 2e^{-\lambda s}\psi_sx_s\sigma_sdW_s + \int_0^t e^{-\lambda s}\psi_s\left(K(1+x_s^2) - \lambda x_s^2 \right)ds.
\end{equation*} 
Let $T^*_R$ denote the exist time of the process from a ball of radius $R$. Since stochastic integrals are local martingales we have
\begin{equation*}
\E \one_A \psi_{T^*}  \leq  \E \one_A  \psi_0  + \E \one_A  \int_0^{T^*} e^{-\lambda s}\psi_s\left(K(1+x_s^2) - \lambda  x_s^2 \right)ds =:\E \one_A \psi_0 + I
\end{equation*} 
 for any stopping time $T^* \leq T^*_R$ and any $A\in \F_0$. But letting $R\to \infty$ we get the above inequality for any ${T^*} \leq T$. We see that $I$ is less then or equal to
\begin{equation*}
\begin{split}
\E \one_A  \Bigg[  \int_0^{T^*} \one_{\{x_s^2<1\}}e^{-\lambda s}\psi_s\left(K(1+x_s^2) - \lambda x_s^2 \right)ds
+ \int_0^{T^*} \one_{\{x_s^2 \geq 1\}}e^{-\lambda s}\psi_s\left(K(1+x_s^2) - \lambda x_s^2 \right)ds \Bigg].
\end{split}
\end{equation*} 
  Choose $\lambda > 0$
  large such that, $2K \leq \lambda$. Then
  \begin{equation*}
      I \leq \E \one_A \int_0^{T^*}
      \one_{\{x_s^2<1\}}e^{-\lambda
        s}\exp(e^{-\lambda
        s}x_s^2)\left(K(1+x_s^2) - \lambda
        x_s^2 \right)ds \leq \E \one_A.
  \end{equation*}  
Hence for any stopping time ${T^*} \leq T$ we have
\begin{equation*} 
\E \one_A \exp(e^{-\lambda {T^*}}x_{T^*}^2) \leq \E \one_A (1 +  \exp(\xi^2)).
\end{equation*} 
Then, due to Gy\"ongy and Krylov \cite[Lemma 3.2]{gyongy:krylov:on}, for any $\delta \in (0,1)$
\begin{equation*}
    \E \sup_{t \in [0,T]}\exp(\delta
    e^{-\lambda t}x_t^2) \leq \tfrac{2-\delta}{1-\delta}(1
    + \E \exp(\delta \xi^2)).
  \end{equation*}
Take $\delta = 1/2$. Let $\mu = \exp(-\lambda T)/2$. Then we see that 
 \begin{equation*}
    \E \sup_{t \in [0,T]}\exp(\mu x_t^2) \leq 3\left(1
    + \E \exp\left(\frac{1}{2} \xi^2\right)\right)
  \end{equation*}
and this completes the proof.
\end{proof}

\begin{corollary}\label{corollary-ptaur}
Let $T^*_R := \inf\{t\geq 0:|x_t| \geq R\}$. If the assumptions
of Lemma~\ref{lemma-exp-est} are satisfied then there exists $\mu >0$ such that
\begin{equation*}
\P(T^*_R \leq T) \leq 3e^{-\mu R^2}(1+\E e^{\xi^2/2}).
\end{equation*}
\end{corollary}

\begin{proof} 
  We will use Lemma~\ref{lemma-exp-est}
  and Markov's inequality.
  \begin{equation*}
    \begin{split}
      \P(T^*_R < T) & = \P \left(\sup_{t
          \in [0,T]}|x_t| \geq R \right) = \P
      \left(\sup_{t \in [0,T]} \exp(\mu
        x_t^2) \geq \exp(\mu R^2) \right)\\ &
      \leq e^{-\mu R^2} \E \sup_{t \in
        [0,T]} e^{\mu x_t^2} \leq 3 e^{-\mu
        R^2}(1+\E e^{\xi^2/2}),
    \end{split}
  \end{equation*}
  where $\mu>0$ comes from~\ref{lemma-exp-est}.
\end{proof}

\begin{lemma}                                                \label{lemma-wbar}
Let $w$ be given by \eqref{eq:am:op:pr2}, let $g$ be a bounded Lipschitz continuous function of $x$ and let Assumption~\ref{ass:bnd} be satisfied. 
Furthermore, let $R_1 > 0$ be given. 
Let $g_{R_1}$ be a function that is equal to $g$ inside $B_{R_1}$, zero outside $B_{R_1+1}$ and Lipschitz continuous.
Let
\begin{equation}                                           \label{eq:wb}
\wb(t,x) := \E_{t,x} \sup_{T^* \in \STtT} \left(e^{-\int_t^{T^*} \rho(u,x_u)du}\gb(x_{T^*}) \right),
\end{equation} 
and let $T^*_{R_1} := \inf\{u \in [t, T] : |x_u^{t,x}| \geq R_1\}$.
Then 
\begin{equation*}
|w(t,x) - \wb(t,x)| \leq K\P_{t,x} \{T^*_{R_1} < T\}.
\end{equation*} 
\end{lemma}

\begin{proof}
Since the difference of supremums is less than the supremum of the difference, 
\begin{equation*}
I:=|w(t,x) - \wb(t,x)| \leq \sup_{T^* \in \STtT}\E_{t,x}\left(e^{\int_t^{T^*} \rho(u,x_u) du} |g(x_{T^*}) - \gb(x_{T^*})|\right).
\end{equation*}
Notice that for $x \in B_{R_1}$, $g(x) = \gb(x)$ and that for all $x\in \R^d$, $|g(x)-\gb(x)|\leq K$. 
Hence
\begin{equation*}
I \leq  \sup_{T^* \in \STtT}\E_{t,x} \one_{\{T^*_{R_1} < T\}} \left(e^{-\int_t^{T^*} \rho(u,x_u) du}K\right).
\end{equation*}
Noting that $\rho \geq 0$ concludes the proof.
\end{proof}
The following corollary is an immediate consequence of Lemma~\ref{lemma-wbar} and
Corollary~\ref{corollary-ptaur}.
\begin{corollary}                                      \label{c:R}
  Let $w$ be given by \eqref{eq:am:op:pr2} and
  $\wb$ be given by \eqref{eq:wb}. Let $R_2 < R_1$ be a positive real number. 
Then there exists $\mu > 0$, such that for all $(t,x)\in [0,T] \times B_{R_2}$,
\begin{equation*}
|w(t,x) - \wb(t,x)| \leq Ce^{-\mu R_1^2 + R_2^2/2}.
\end{equation*}
\end{corollary}

In other words, by cutting the function $g$ outside the ball of radius $R_1$ we have introduced an error in the American put price that is decreasing exponentially as the ball increases in radius. 
Now we will show that, if the discrete problem is solved with a payoff function, that is zero outside a ball of radius $R_1$ and the grid is restricted to a ball of radius $R>R_1$, then the difference between this solution and the solution of the discrete problem on the whole grid with the original payoff function decreases exponentially with $R$ and $R_1$.
\begin{lemma} 
\label{lemma-discrete-comparison-principle}
Let $g_1$ and $g_2$ be functions of $x$, such that $g_1 \leq g_2 < \infty$.
Let $u_1,u_2$ be defined on $\bar{\mathcal{M}}_T$ and such that $u_1e^{-\mu|x|}$ and $u_2 e^{-\mu|x|}$ are bounded for some $\mu>0$. 
Let $C \geq 0$ and assume that
\begin{equation}
\label{eqn-comp-princ-ass1}
\max\left[\fddeltatau^T u_1 + \L_h u_1 + C, g-u_1 \right] \geq \max\left[\fddeltatau^T u_2 + \L_h u_2, g-u_2\right] \on Q
\end{equation}
and $u_1 \leq u_2$ on $\bmt \setminus Q$.
If $h\leq 1$ then there is a constant $\bar{\tau}$ depending only on $K, d_1, \mu$ such that
\begin{equation}\label{eqn-comparison-princ-conclusion}
u_1 \leq u_2 + TC \on \bar{\mathcal{M}}_T,
\end{equation}
for $\tau \in (0, \bar{\tau})$. If $u_1$, $u_2$ are bounded on
  $\mathcal{M}_T$ then
  \eqref{eqn-comparison-princ-conclusion}
  holds for any $\tau$.
\end{lemma}
Lemma~\ref{lemma-discrete-comparison-principle} is a special case of \cite[Lemma 3.9]{gyongy:siska:on:the:rate}, while following corollary is just a special case of 
\cite[Corollary 3.11]{gyongy:siska:on:the:rate}.
\begin{corollary}
\label{corollary-soln-to-disc-bellman-pde-is-bdd}
If $\wth$ is the solution of \eqref{eq:fd:scheme:american} then 
$$|\wth| \leq C+\sup_{x\in \R^d}|g(x)|.$$
\end{corollary}

\begin{lemma}
\label{lemma-R-est}
Assume that $g(x) = 0$ for $|x| \geq R.$ 
Then, for some $\gamma \in (0,1)$,
$$|\wth(t,x)| \leq C_Te^{\gamma (R-|x|)}.$$
\end{lemma}
\begin{proof} 
Let $\gamma \in (0,1)$ be some constant to be chosen later.
Let $\xi(t)$ be defined recursively by:
\begin{equation*}
\xi(T) = 1; \quad \xi(t) = \gamma^{-1}\xi(t+\tau_T(t)) \cfor t < T.
\end{equation*} 
Take an arbitrary unit $l \in \R^d$. 
Let $\eta(x) = \exp(\gamma (x,l))$ and $\zeta = \xi \eta$.
We would like to apply Lemma~\ref{lemma-discrete-comparison-principle}. Observe
  that by Taylor's theorem
\begin{equation*}
\begin{split}
\fdDeltak \eta(y) & = \Dlk^2 \eta(y) + \frac{1}{6h^2} \int_{-h}^h \Dlk^4 \eta(y+sl_k)(h-|s|)^3ds \\ 
& \leq \Dlk^2 \eta(y) + \frac{h^2}{12}\sup_{s\in (-h,h)}|\Dlk^4 \eta(y+sl_k)|.
\end{split}
\end{equation*} 
As $|l_k| < K$ and $\gamma \in (0,1)$, we obtain that
\begin{equation*}
\Dlk^2 \eta = \eta \gamma^2 (l,l_k)^2 \leq C \gamma \eta \cand \Dlk^4 \eta \leq C \gamma \eta.
\end{equation*}
Hence
\begin{equation*}
\sup_{s\in (-h,h)}|\Dlk^4 \eta(y+sl_k)| \leq C \gamma \eta(y)\sup_{s\in (-h,h)}|\exp(sl_k,l_k)| \leq C \gamma \eta(y),
\end{equation*} 
provided $h < K$. 
Thus $\fdDeltak \eta \leq C \gamma \eta$.
Similarly $\fddeltak \eta \leq C \gamma \eta$.
Furthermore, since $\fddeltatau^T(\xi \eta) = \tau^{-1}(\gamma - 1)\xi \eta$, we can see that, for sufficiently small $\gamma$
  \begin{equation*}
    \tfrac{1}{1+r}(\fddeltatau^T \zeta + \L_h
    \zeta - r \zeta) \leq (\tau^{-1}(\gamma - 1) +
    C\gamma) \zeta \leq 0 \on H_T.
  \end{equation*} Let $Q = \{(t,x) \in
  \mathcal{M}_T:(x,l) \leq -R\}$. Since
  $-R \geq (x,l) = |x|\cos \theta \geq
  -|x|$, $Q \subset \{(t,x) \in
  \mathcal{M}_T: |x| \geq R\}$. Recall
  that $g(x) = 0$ for $|x| \geq
  R$. Hence for any constant $C>0$,
  $$ \tfrac{1}{1+r}(\fddeltatau^T C\zeta + \L_h C\zeta -rC\zeta + rg) \leq 0 \on Q.$$
  For $(t,x) \in \bar{\mathcal{M}}_T
  \setminus Q$, either $t < T$ and $(x,l)
  > - R$ or $t = T$ and $(x,l) \geq
  -|x|$. Hence either
  \begin{equation*}
    \zeta(t,x) \geq e^{-\gamma R} \textrm{ or }
    \zeta(t,x) \geq e^{-\gamma |x|}.
  \end{equation*} 
In the first case, we know from Corollary~\ref{corollary-soln-to-disc-bellman-pde-is-bdd} that $\wth$ is bounded by a constnat.
Taking $C$ large enough, $e^{\gamma R}C\zeta \geq \wth$.
In the second
  case, $t = T$ and so $\wth =
  g$. As $g(x)=0$ for $|x|\geq R$ we only
  need to consider $|x| < R$ and so for
  large $C$, $e^{\gamma R}C\zeta \geq g =
  \wth$.
Either way, for $C$ large enough, $Ce^{\gamma R}\zeta \geq \wth$ on $\bar{\mathcal{M}}_T \setminus Q$.
By
  Lemma~\ref{lemma-discrete-comparison-principle}
  $\wth \leq Ce^{\gamma R}\zeta$ in
  $\mathcal{M}_T$. Since the choice of the
  unit vector $l$ was arbitrary we can see
  that in $\mathcal{M}_T$
  \begin{equation*}
    \wth \leq C_Te^{\gamma
      R}\exp(-\gamma |x|).
  \end{equation*} Analogous application of
  Lemma~\ref{lemma-discrete-comparison-principle}
  would yield that in $\mathcal{M}_T$
  \begin{equation*}
    \wth \geq -C_Te^{\gamma
      R}\exp(-\gamma |x|)
  \end{equation*}
  and hence complete the proof.
\end{proof}

\begin{lemma}\label{lemma-fd-to-balls}
Let Assumption~\ref{ass:bnd} and~\ref{ass-on-the-scheme} be satisfied.
Let $R\geq R_1>0$ and consider $\gb$, a bounded Lipschitz continuous function that is zero outside $B_{R_1}$.
Let $\wthRone$ be a function satisfying \eqref{eq:fd:scheme:american} on $\mt$ and such that $\wthRone = \gb$ on $\bmt \setminus \mt$.
Let $\wthRR$ be the solution to \eqref{eq:fd:scheme:american} on $Q_R$ with $\wthRR = \gb$ on $\bmt \setminus Q_R$. 
Then for some constant $\gamma \in (0,1)$,
$$|\wthRone - \wthRR| \leq Ce^{\gamma(R_1 - R)} \on \bmt.$$
\end{lemma}

\begin{proof} 
From Lemma~\ref{lemma-R-est} we know that for $|x| \geq R$, $|\wthRone| \leq Ce^{\gamma(R_1 - R)}$. 
Since $\wthRone$ satisfies \eqref{eq:fd:scheme:american} on $\mt$ and $\wthRR$ satisfies \eqref{eq:fd:scheme:american} on $Q_R$, we have, on $Q_R$,
\begin{equation*}
\begin{split} 
& \max\left(\fddeltatau^T \wthRR + L_h \wthRR, g-\wthRR\right) = 0 = \max\left(\fddeltatau^T \wthRone + \L_h \wthRone, g-\wthRone\right)\\
& \leq \max\left(\fddeltatau^T (\wthRone-Ce^{\gamma(R_1-R)}) + \L_h (\wthRone-Ce^{\gamma(R_1-R)}), g - (\wthRone - Ce^{\gamma(R_1 - R)})\right). \\ 
\end{split}
\end{equation*} 
For $(t,x) \in \bmt \setminus Q_R$ either $(t,x) \in [0,T)\times B_R^c \cap \mt$ and so
$$\wthRone - Ce^{\gamma(R_1 - R)} \leq 0 = \gb = \wthRR$$
or $t = T$ and so
$$\wthRone - Ce^{\gamma(R_1 - R)} = \gb - Ce^{\gamma(R_1 - R)} \leq \gb = \wthRR.$$ 
Hence by Lemma~\ref{lemma-discrete-comparison-principle}
applied to $\wthRone - Ce^{\gamma(R_1-R)}$ and $\wthRR$ on $Q_R$
$$\wthRone \leq \wthRR +  Ce^{\gamma(R_1 - R)} \on \bmt.$$ 
Similar argument for $\wthRone + Ce^{\gamma (R_1 - R)}$ gives
$$\wthRone + Ce^{\gamma(R_1 - R)} \geq \wthRR \on \bmt$$
and so completes the proof.
\end{proof}

\begin{proof}[Proof of Theorem~\ref{thm:american}.]
By Lemma~\ref{l:u}, \eqref{eq:fd:scheme:american} has a unique solution.
We consider the optimal stopping problem $w_{R_1}$ given by \eqref{eq:am:op:pr2}, but with $g$ replaced by $g_{R_1}$.
By Corollary~\ref{c:R} there exists $\mu > 0$ such that
\begin{equation*}
|w - w_{R_1}| \leq C e^{-\mu R_1^2 +  R_2^2/2} \on [0,T] \times B_{R_2}.
\end{equation*}
Let $\wthRone$ be the solution of \eqref{eq:fd:scheme:american} on $\mt$ with $g$ replaced by $g_{R_1}$.
This is the solution of the finite difference problem with $g$ cut off outside a ball, but on a grid on the whole space.
By Theorem~\ref{t:rate}, we know that
\begin{equation*}
|w_{R_1} - \wthRone| \leq C(\tau^{1/4} + h^{1/2}) \on [0,T]\times \R^d.
\end{equation*}
Finally, we consider $\wthRR$, which is the solution of \eqref{eq:fd:scheme:american} on $Q_R$ and with $g$ replaced by $g_{R_1}$.
By Lemma~\ref{lemma-fd-to-balls} there exists $\gamma \in (0,1)$, such that
\begin{equation*}
|\wthRone - \wthRR| \leq Ce^{\gamma(R_1  - R)} \on [0,T] \times B_{R_1}.
\end{equation*}
Hence, chaining the above three inequalities, we obtain
\begin{equation*}
|w - \wthRR| \leq C\left(\tau^{1/4} + h^{1/2} + e^{\gamma(R_1  - R)} +  e^{-\mu R_1^2 + R_2^2/2}\right),
\end{equation*}
which is the estimate for the error caused by both discretisation and localisation.
\end{proof}

\paragraph{\bf Acknowledgements} The author is grateful 
to Etienne Emmrich and Istv\'an Gy{\"o}ngy for pointing 
out a number of mistakes in the drafts and for many useful suggestions.

\bibliographystyle{cmam}

\bibliography{bibliography}

\end{document}